\documentclass[12pt]{amsart}

\usepackage{amsfonts}
\usepackage{amssymb}
\usepackage{amsmath}
\usepackage{amsthm}
\usepackage{mathrsfs}
\usepackage{enumerate}
\usepackage{dsfont}
\usepackage{yfonts}
\usepackage{mathbbol}

\usepackage[arrow, matrix, curve]{xy}
\xymatrixcolsep{0.5cm}
\xymatrixrowsep{0.5cm}
\theoremstyle{plain}

\newtheorem{condition}{Condition}

\newtheorem{Def}[condition]{Definition}

\newtheorem{lemma}[condition]{Lemma}

\newtheorem{prop}[condition]{Proposition}
\newtheorem{theorem}[condition]{Theorem}
\theoremstyle{remark}
\newtheorem{remark}[condition]{Remark}

\numberwithin{equation}{section}
\numberwithin{condition}{section}

\DeclareMathOperator{\ad}{ad}

\DeclareMathOperator*{\slim}{s-lim}

\DeclareMathOperator*{\supp}{supp}
\DeclareMathOperator{\ran}{ran}

\title[Regularity of eigenstates]{Regularity of eigenstates in regular Mourre-theory}
\author{Jacob S. M\o ller}
\address[J. M\o ller and M. Westrich]
{Dept. of Mathematics,\newline%
\indent Aarhus Universitet}%
\email[J. M\o ller]{jacob@imf.au.dk}%
\author{Matthias Westrich}
\email[M. Westrich]{westrich@imf.au.dk}%
\date{02.07.2010}
\begin{document}

\begin{abstract}
The present paper gives an abstract method to prove that possibly embedded eigenstates of a self-adjoint operator $H$ lie in the domain of the $k^{th}$ power of a conjugate operator $A$. Conjugate means here that $H$ and $A$ have a positive commutator locally near the relevant eigenvalue in the sense of Mourre. The only requirement is $C^{k+1}(A)$ regularity of $H$. Regarding integer $k$, our result is optimal. Under a natural boundedness assumption of the multiple commutators we prove that the eigenstate 'dilated' by $\exp(i\theta A)$ is analytic in a strip around the real axis. In particular, the eigenstate is an analytic vector with respect to $A$. Natural applications are 'dilation analytic' systems satisfying a Mourre estimate, where our result can be viewed as an abstract version of a theorem due to Balslev and Combes, \cite{BalslevCombes1971}. As a new application we consider the massive Spin-Boson Model.
\end{abstract}
\maketitle

\section{Introduction and main results}

In this paper we study regularity of eigenstates $\psi$
of a self-adjoint operator $H$, with respect to an auxiliary
operator $A$ for which $i[H,A]$ satisfies a so-called Mourre estimate near the
associated eigenvalue $\lambda$. Our results are partly an extract of
a recent work of Faupin, Skibsted and one of us \cite{FaupinMoellerSkibsted2010a}, 
and partly an improvement of a result of Cattaneo, 
Graf and Hunziker \cite{CattaneoGrafHunziker2006}.  We consider in the present work the case of regular Mourre theory, where the derivation of the bounds on $A^k\psi$ is simpler compared to \cite{FaupinMoellerSkibsted2010a}. In fact we derive explicit bounds which are independent of proof technical constructions. The bounds are good enough to formulate a reasonable condition on the growth of norms of multiple commutators which ensures that eigenstates are analytic vectors with respect to A. We discuss how these growth conditions may be checked in concrete examples and illustrate this for dilation analytic $N$-body Hamitonians and the massive Spin-Boson Model.

The general strategy in this paper, as well as in \cite{CattaneoGrafHunziker2006} and \cite{FaupinMoellerSkibsted2010a},
is to implement a Froese-Herbst type argument in an abstract setting.
In a formal computation the Mourre estimate suffices to extract
results of the type presented here but to make the argument 
rigorous one has to impose
enough conditions on the pair of operators $H$ and $A$ to enable a 
calculus of operators. This is usually done by requiring
a number of iterated commutators between $H$ and $A$ to exist
and be controlled by
operators already present in the calculus. The type of conditions imposed 
is typically guided by a set of applications that the authors have in mind.
Most examples, like many-body quantum systems with or without external 
classical fields, have been possible to treat using natural extensions of 
conditions originally introduced by Mourre in \cite{Mourre1981}. The same goes for
a number of models in non-relativistic QED like confined 
massive Pauli-Fierz models and massless models, with $A$ being the generator 
of dilations. These are the type of conditions used in \cite{CattaneoGrafHunziker2006}.

Over the last 10 years a number of models that fall outside the scope of
Mourre's original conditions, and hence not covered by \cite{CattaneoGrafHunziker2006},
have appeared. We split them in two types. The first type
are models that, while not covered by Mourre type conditions on iterated 
commutators, still satisfy weaker conditions developed  over some years by
Amrein, Boutet de Monvel, Georgescu and Sahbani \cite{AmreinBoutetdeMonvelGeorgescu1996,Sahbani1997}. 
These conditions play the same role as Mourre's original conditions in 
that they enable the same type of calculus of the operators $H$ and $A$.
We call this setting for \emph{regular Mourre theory}. Examples of models
that fall in this category but are not covered by Mourre type conditions 
as in \cite{CattaneoGrafHunziker2006}, are: $P(\phi)_2$-models \cite{DerezinskiGerard2000} 
(with $P(\varphi)\neq \varphi^4$), 
the renormalised massive Nelson model \cite{Ammari2000},
Pauli-Fierz type models without confining potential \cite{FroehlichGriesemerSchlein2001}, the standard model of non-relativistic QED near the ground state energy, where only local $C^k$ conditions are available, \cite{FroehlichGriesemerSigal2008},
and the translation invariant massive Nelson model \cite{MoellerRasmussen2010a}.

The second type of models we wish to highlight are those for
which the commutator $H'= i[H,A]$ is not comparable to $H$ (or $A$).
Here one views the commutator as a new operator in the calculus
and impose assumptions of mixed iterated commutators between the three
possibly unbounded operators $H,A$ and $H'$. This type of analysis goes back
to \cite{HuebnerSpohn1995b,Skibsted1998} and was further developed in \cite{MoellerSkibsted2004} and \cite{GeorgescuGerardMoeller2004a}.
This situation we call \emph{singular Mourre theory} 
and is the topic considered in \cite{FaupinMoellerSkibsted2010a}. There are two examples where this
type of analysis is natural. The first is massless Pauli-Fierz models
with $A$ being the generator of radial translations 
\cite{DerezinskiJaksic2001,GeorgescuGerardMoeller2004b,FaupinMoellerSkibsted2010a,FaupinMoellerSkibsted2010b,Skibsted1998,Golenia2009} and the second is many-body systems 
with time-periodic pair-potentials, in particular AC-Stark Hamiltonians, 
\cite{MoellerSkibsted2004,FaupinMoellerSkibsted2010a}. The technical complications arising from having to deal 
with a calculus of three unbounded operators are significant.

Part of the motivation of this work is  to extract the essence of \cite{FaupinMoellerSkibsted2010a}
in the context of regular Mourre theory, where the technical overhead is more
manageable.

A second motivating factor is drawn from the paper \cite{FaupinMoellerSkibsted2010b},
which is in fact intimately connected to \cite{FaupinMoellerSkibsted2010a}. We remind the reader 
of the Fermi Golden Rule (FGR) which we now formulate. 
Let $P$ denote the orthogonal projection
onto the span of the eigenvector $\psi$, and abbreviate $\bar{P} = I-P$. 
The FGR states that a, for simplicity isolated and simple, 
embedded eigenvalue is unstable under a  perturbation $W$ provided
\begin{equation}\label{FGR}
\Im \lim_{\epsilon\to 0+} 
\langle W\psi,\bar{P}(\bar{H}-\lambda-i\epsilon)^{-1}\bar{P}W\psi\rangle \neq 0.
\end{equation}
Here $\bar{H} = \bar{P} H\bar{P}$ as an operator on the range of $\bar{P}$.
In the above statement the existence of the limit is of course implicitly assumed.
Due to the presence of the projection $\bar{P}$, the operator $\bar{H}$
has purely continuous spectrum near the eigenvalue $\lambda$, 
and the existence of the limit can thus be inferred from the 
\emph{limiting absorption principle} (LAP). 
The LAP can be deduced using positive commutator estimates, see e.g. \cite{AmreinBoutetdeMonvelGeorgescu1996}, 
provided there exists an auxiliary operator $A$ 
such that $H$ and $A$ satisfy a Mourre estimate near $\lambda$ 
and  $(\bar{H}-i)^{-1}$ admits two bounded commutators with $A$, 
or more precisely $H$ is of class $C^2(\bar{P}A\bar{P})$ 
(see the next subsection).
This implies in particular that $\ran(P)\subseteq\mathcal{D}(A^2)$, i.e. 
$\psi\in \mathcal{D}(A^2)$.
Even by the improvement of \cite{FaupinMoellerSkibsted2010a}, and in turn this paper, we would still
need $H$ to be of class $C^3(A)$ in order to verify this property. 
This would for example preclude application to the model considered in 
\cite{MoellerRasmussen2010a}. In \cite{FaupinMoellerSkibsted2010b} the authors study the limit in 
\eqref{FGR} directly, bypassing the general limiting absorption theorems, 
albeit applying the same differential inequality technique, and prove 
existence of the limit assuming only $\psi\in \mathcal{D}(A)$. 
Combined with \cite{FaupinMoellerSkibsted2010a} (or this paper)
this establishes the existence of the limit in the 
Fermi Golden Rule \cite{FaupinMoellerSkibsted2010b} abstractly under a $C^2(A)$ condition.
The price to pay is that one needs a prior control of the norm $\|A\psi\|$
locally uniformly in possibly existing perturbed eigenstates. While it is clear
that such a locally uniform bound does hold, 
provided all the input in \cite{FaupinMoellerSkibsted2010a}
is controlled locally uniformly in the perturbation, it is however impractical
due to the complexity of the setup to extract such bounds in closed form.
In this paper we do just that in the simpler context of regular Mourre theory.

As a last motivation, we had in mind a consequence of having good 
explicit bounds on the norms $\|A^k\psi\|$. Namely, provided one
imposes natural conditions on the norms of all iterated commutators,
we show as a consequence of our explicit bounds on $\|A^k\psi\|$
that the power series $\sum_{k=1} \frac{(i\theta A)^k}{k!}\psi$ has a 
positive radius of convergence, thus establishing that 
$\psi$ is an analytic vector for $A$.
Here however, we have to work with conditions of the type considered 
in \cite{CattaneoGrafHunziker2006}. Having established analyticity of the map $\theta\mapsto \exp(i\theta A)\psi$ in a ball around $0$ one may observe that this map is actually analytic in a strip around the real axis, and thus this result reproduces a result of Balslev and Combes, \cite[Thm.1]{BalslevCombes1971} on analyticity of dilated non-threshold eigenstates. As an example of a new result, we prove for the massive Spin-Boson Model that non-threshold eigenstates are analytic vectors with respect to the second quantised generator of dilations.


\subsection{Commutator Calculus}

We pause to introduce the commutator calculus of \cite{AmreinBoutetdeMonvelGeorgescu1996} 
before formulating our main results. Let $A$ be a self-adjoint operator with domain $\mathcal{D}(A)$ in a Hilbert space $\mathcal{H}$. We denote with $\mathfrak{B}(X,Y)$ the set of bounded operators on the normed space $X$ with images in the normed space $Y$ and $\mathfrak{B}(X):=\mathfrak{B}(X,X)$.
\begin{Def}
A bounded operator $B\in\mathfrak{B}(\mathcal{H})$ is said to be of class $C^k(A)$, in short $B\in C^k(A)$, if
\begin{equation}
\label{eq:Ckdef}
\mathds{R}\ni t\mapsto e^{itA}Be^{-itA}
\end{equation}is strongly in $C^k(\mathds{R})$. A, possibly unbounded self-adjoint operator $S$ is said to be of class $C^k(A)$ if $(i-S)^{-1}\in C^k(A)$.
\end{Def}The property, that $B\in\mathfrak{B}(\mathcal{H})$ is of class $C^1(A)$ is equivalent to the statement that
\begin{equation*}
(\phi,[B,A]\chi):=(B^*\phi,A\chi)-(A\phi,B\chi),\,\,\forall \phi,\chi\in \mathcal{D}(A)
\end{equation*}extends to a bounded form on $\mathcal{H}\times\mathcal{H}$, which in turn is implemented by a bounded operator, $\ad_A(B)$, see e.g. \cite{GeorgescuGerardMoeller2004b}. If $B\in C^2(A)$, then an argument using Duhamel's formula shows $\ad_A(B)\in C^1(A)$ and thus there exists a bounded extension of the form $[\ad_A(B),A]$. Thus, one constructs for $B\in C^k(A)$ iteratively the bounded operator $\ad_A^k(B):=\ad_A(\ad_A^{(k-1)}(B))$. We set $\ad^0_A(B):=B$.

Commutators involving two possibly unbounded self-adjoint operators $H$ and $A$ will in general not extend to bounded operators on $\mathcal{H}$ and the definition of the quadratic form $[H,A]$ requires further restrictions on its domain. Thus we denote by $[H,A]$ the form
\begin{equation*}
(\phi,[H,A]\chi):=(H\phi,A\chi)-(A\phi,H\chi),\,\,\forall \phi,\chi\in \mathcal{D}(A)\cap \mathcal{D}(H).
\end{equation*}If $H\in C^1(A)$, then $\mathcal{D}(A)\cap \mathcal{D}(H)$ is dense in $\mathcal{D}(H)$ in the graph norm of $H$ and $[H,A]$ extends to a $H$-form bounded quadratic form, which in turn defines an unique element of $\mathfrak{B}(\mathcal{D}(H),\mathcal{D}(H)^*)$ denoted by
\begin{equation*}
\ad_A(H):\mathcal{D}(H)\to \mathcal{D}(H)^*,
\end{equation*}see \cite{GeorgescuGerardMoeller2004a}. The space  $\mathcal{D}(H)^*$ is the dual of $\mathcal{D}(H)$ in the sense of rigged Hilbert spaces.
\\

Our result on the analyticity of eigenvectors of $H$ with respect to $A$ requires a construction of multiple commutators of $H$ and $A$ which are bounded as maps from $\mathcal{D}(H)$ to $\mathcal{H}$ in the graph norm of $H$. The construction is as follows: Let $H\in C^1(A)$. We assume that $\ad_A(H)\in\mathfrak{B}(\mathcal{D}(H),\mathcal{H})$. Then, $[\ad_A(H),A]$ is defined as
\begin{equation}
(\psi,[\ad_A(H),A]\phi):=(-\ad_A(H)\psi,A\phi)-(A\psi,\ad_A(H)\phi),
\end{equation}for all $\psi,\phi\in\mathcal{D}(A)\cap\mathcal{D}(H)$. Here we used, that $\ad_A(H)$ is skew-symmetric on the domain $\mathcal{D}(A)\cap\mathcal{D}(H)$. Assume that this form extends in graph norm of $H$ to a form which is implemented by an element $\ad_A^2(H)\in\mathfrak{B}(\mathcal{D}(H),\mathcal{H})$. Proceeding iteratively, we construct $\ad_A^k(H)\in\mathfrak{B}(\mathcal{D}(H),\mathcal{H})$.
\begin{lemma}
\label{lemmaCkproperty}
Let $H,A$ be self-adjoint operators on the Hilbert space $\mathcal{H}$ and assume $H\in C^1(A)$. If $\ad_A^j(H)\in\mathfrak{B}(\mathcal{D}(H),\mathcal{H})$ for $0\leq j \leq k$, then $H\in C^k(A)$.
\end{lemma}The proof of this lemma may be found in Section \ref{sectanalyticity}.\\

In several places we need an appropriate class of functions to regularise the self-adjoint operators $H,A$, defined on $\mathcal{D}(H),\mathcal{D}(A)$ respectively, and enable a calculus for them.
\begin{Def}
Define $\mathscr{B}:=\big\{r\in C_b^{\infty}(\mathds{R},\mathds{R})\big|r'(0)=1,\,\,r(0)=0,\,\,\forall k\in\mathds{N}:\,\,\sup_{t\in\mathds{R}}|r^{k}(t)\langle t\rangle^{k}|<\infty,\,\,\textrm{$r$ $\mathrm{is}$ $\mathrm{real}$ $\mathrm{analytic}$ $\mathrm{in}$ $\mathrm{some}$ $\mathrm{ball}$ $\mathrm{around}$ $0$}\big\}$.
\end{Def}

Let $h\in\mathscr{B}$. For $\lambda\neq 0$ redefine $h_{\lambda}(x):=h(x-\lambda)$. In the following we will drop the index $\lambda$ as well as the argument of $h_{\lambda}(H)$ and other regularisations of $H$ and $A$, if the context is clear. The following condition is a local $C^1(A)$ condition, as in \cite{Sahbani1997}, plus a Mourre estimate.
\begin{condition}
Let $H,A$ be self-adjoint operators on $\mathcal{H}$ and $\lambda\in\mathds{R}$. There exists an $h\in \mathscr{B}$, $h_{\lambda}(s):=h(s-\lambda)$, with $h_{\lambda}(H)\in C^1(A)$ and an $f_{\mathrm{loc}}\in C_0^{\infty}(\mathds{R},[0,1])$, such that $f_{\mathrm{loc}}(\lambda)=1$ and $h_{\lambda}'(x)>0$ for all $x\in\supp(f_{\mathrm{loc}})$. Assume there is a \emph{smooth Mourre estimate}, i.e. $\exists C_0,C_1>0$ and a compact operator $K$, such that
\begin{equation}
i\ad_A(h_{\lambda}(H))\geq C_0-C_1f_{\mathrm{loc},\perp}^2(H)-K.
\label{eq:smoothmourre}
\end{equation}$f_{\mathrm{loc},\perp}$ is defined as $f_{\mathrm{loc},\perp}:=1-f_{\mathrm{loc}}$.
\label{condmourre1}
\end{condition}
\begin{remark}
\begin{enumerate}
\item The requirement $h'_{\lambda}(x)>0$, $\forall x\in\supp(f_{\mathrm{loc}})$, implies $f_{\mathrm{loc}}\in C^k(A)$ if $h_{\lambda}\in C^{k}(A)$ for $k\in\mathds{N}$, since $h_{\lambda}$ is smoothly invertible (on each connected component of $\supp(f_{\mathrm{loc}})$) and $f_{\mathrm{loc}}$ may be written as a smooth function of $h_{\lambda}$.
\item The assumption of $K$ being compact is not necessary. In fact we could replace this by the requirement that $\textbf{1}_{|A|\geq\Lambda}K$, where $\textbf{1}_{|A|\geq\Lambda}$ denotes the spectral projection on $[\Lambda,\infty)$, can be made arbitrarily small.
\item For a discussion of the 'local' Mourre estimate (\ref{eq:smoothmourre}) with the standard form of the Mourre estimate see Section \ref{mourre}.
\end{enumerate}
\end{remark}
\begin{theorem}[Finite regularity]
\label{finitereg}
Let $H,A$ be self-adjoint operators on the Hilbert space $\mathcal{H}$ and $\psi$ be an eigenvector of $H$ with eigenvalue $\lambda$. Assume Condition \ref{condmourre1} to be satisfied with respect to $\lambda$ and $h_{\lambda}(H)\in C^{k+1}(A)$ for some $k\in\mathds{N}$. There exists $c_k>0$, only depending on $\supp(f_{\mathrm{loc}})$, $C_0$, $C_1$, $K$, $\|\ad_A^{\ell}(f_{\mathrm{loc}}(H))\|$, $\|\ad_A^j(h_{\lambda}(H))\|$, $1\leq\ell\leq k$, $1\leq j\leq k+1$, such that
\begin{equation}
\left\|A^k\psi\right\|\leq c_k\left\|\psi\right\|.
\label{eq:localuniformbound}
\end{equation}
\end{theorem}
\begin{remark}
In \cite[Ex. 1.4]{FaupinMoellerSkibsted2010a} it is shown, that the statement of Theorem \ref{finitereg} is false in general if one requires $h_{\lambda}\in C^{k}(A)$ only. Therefore, the result is optimal concerning integer values of $k$.
\end{remark}
\begin{condition}The self-adjoint operator $H$ is of class $C^1(A)$ and there exists a $v>0$, such that for all $k\in\mathds{N}$
\begin{equation}
\|\ad^k_A(H)(i-H)^{-1}\|\leq k!v^{-k}.
\label{eq:condad}
\end{equation}
\label{condad}
\end{condition}
\begin{theorem}[Analyticity]
\label{analyticity} 
Let $H,A$ be self-adjoint operators on the Hilbert space $\mathcal{H}$ and $\psi$ be an eigenvector of $H$ with eigenvalue $\lambda$. Assume Condition \ref{condmourre1} to be satisfied with respect to $\lambda$ and that Condition \ref{condad} holds. Then, the map
\begin{equation}
\mathds{R}\ni\theta\mapsto e^{i\theta A}\psi\in\mathcal{H}
\label{eq:stripana}
\end{equation}extends to an analytic function in a strip around the real axis.
\end{theorem}

\section{Applications}
\label{applications}
The applications of our result on "finite regularity of eigenstates" are well known and discussed in the literature \cite{AgmonHerbstSkibsted1989,CattaneoGrafHunziker2006,HunzikerSigal2000b,MoellerSkibsted2004,FaupinMoellerSkibsted2010b}. In contrast results on the analyticity of eigenvalues in regular Mourre theory are to our knowledge unknown. Even though the condition under which our result holds appears difficult to verify in concrete situations, we will illustrate for some deformation analytic models that it is strikingly simple to check the assumptions of Theorem \ref{analyticity}.\\
Let $H$ be a self-adjoint operator on the Hilbert space $\mathcal{H}$ and 
$U(t):=\exp(itA)$ a strongly continuous one parameter group of unitary operators $U(t)$. 
The self-adjoint operator $A$ is the generator of this group. 
Assume that $U(t)$ \emph{b-preserves} $\mathcal{D}(H)$, i.e. a
\begin{equation*}
U(t)\mathcal{D}(H)\subseteq\mathcal{D}(H),\,\,\forall t\in\mathds{R}
\ \textup{and} \ \sup_{t\in[-1,1]}\|U(t)\phi\|_{\mathcal{D}(H)}<\infty,\,\,\forall\phi\in\mathcal{D}(H),
\end{equation*}
where $\|\psi\|_{\mathcal{D}(H)}$ denotes the graph norm of $H$. 

\begin{remark}\label{Rem-b-preserve} Observe that the following are equivalent:
\begin{itemize}
\item $U(t)$ b-preserves $\mathcal{D}(H)$.
\item There exists $\mu_0>0$ and $C>0$  such that for all $\mu\in\mathbb{R}$ with $|\mu|\geq \mu_0$,
we have $(A-i\mu)^{-1}:\mathcal{D}(H)\to\mathcal{D}(H)$ and
\[
\|(A-i\mu)^{-1}\|_{\mathfrak{B}(\mathcal{D}(H),\mathcal{H})}\leq C|\mu|^{-1}.
\]
\end{itemize}
\end{remark}

By \cite[Lemma 2.33]{GeorgescuGerardMoeller2004a} one observes that 
$U^{\circ}(\cdot):=U(\cdot)_{|\mathcal{D}(H)}$ 
is a $C_0$-group in the topology of $\mathcal{D}(H)$.

\begin{prop}\label{bdcommutatersandCk}
Let $H,A$ be self-adjoint operators and $U(t):=\exp(itA)$.
Assume that $U(\cdot)$ b-preserves $\mathcal{D}(H)$. 
Then for any $k\in\mathds{N}$ the following statements are equivalent.
\begin{enumerate}
\item $H$ admits $k$ $H$-bounded commutators with $A$, denoted by $\ad^j_A(H)$, $j=1,\dots,k$.
\item The map $t\mapsto I(t)=(\varphi,U(t)HU(t)^*\psi)\in C^k([-1,1])$, for all 
$\psi,\varphi\in\mathcal{D}(H)\cap\mathcal{D}(A)$. There exist $H$-bounded operators $H^{(j)}(0)$, $j=1,\dots,k$, 
such that  $\frac{d^j}{d t^j}I(t)_{|t=0} = (\varphi,H^{(j)}(0)\psi)$, 
for $j=1,\dots,k$ and all $\psi,\varphi\in\mathcal{D}(H)\cap\mathcal{D}(A)$.
\item $t\mapsto \psi(t):=U(t)HU(t)^*\psi \in C^k([-1,1];\mathcal{H})$ for all $\psi\in\mathcal{D}(H)$, and there exist  $H$-bounded operators
$H^{(j)}(0)$, $j=1,\dots,k$, with the property that $\frac{d^j}{d t^j} \psi(t)_{|t=0} = H^{(j)}(0)\psi$, for 
all $j=1,\dots,k$ and $\psi\in \mathcal{D}(H)$.  
\end{enumerate}
If one of the three statements holds, then the pertaining $H$-bounded operators are uniquely determined
and we have
\begin{equation}
i^{j} \ad_A^{j}(H)= (-1)^{j} H^{(j)}(0),\,\,j=1,\dots,k.
\label{eq:H(k)commutators}
\end{equation}
\end{prop}

\begin{proof} Assume the commutator form $[H,A]$ has an extension from
$\mathcal{D}(H)\cap\mathcal{D}(A)$ to an $H$-bounded operator.
Then an argument of Mourre, \cite[Prop.II.2]{Mourre1981}, keeping Remark~\ref{Rem-b-preserve} in mind, 
implies that $(H+i)^{-1}:\mathcal{D}(A)\to\mathcal{D}(A)$.
Hence, it follows that $(H+i)^{-1}$ is of class $C^1(A)$. A consequence of this is that
$\mathcal{D}(A)\cap\mathcal{D}(H)$ is dense in $\mathcal{D}(H)$ (as well as in $\mathcal{D}(A)$).
(Alternatively use Remark~\ref{Rem-b-preserve} backwards in conjunction with Nelson's theorem, \cite[Thm. X.49]{ReedSimonII1975}.)
This remark implies that any extension of the commutator form $[H,A]$ to an $H$-bounded operator
is necessarily unique.

(1) $\Rightarrow$ (2):  A consequence of the above observation 
is that $\ad_A^j(H)$, for $j=1,\dots,k$, is symmetric for $j$ even and anti-symmetric for $j$ odd.
Compute first for $\varphi, \psi\in\mathcal{D}(H)\cap\mathcal{D}(A)$
\[
\frac{d}{dt} I(t) = -(\varphi,U(t)i[H,A]U(t)^*\psi) = -(\varphi,U(t)i\ad_A(H)U(t)^*\psi).
\]
If we evaluate at $t=0$ we observe that $H^{(1)}(0) = - i \ad_A(H)$ can be used as a weak 
derivative on $\mathcal{D}(H)\cap\mathcal{D}(A)$. Iteratively we now conclude that
\[
\frac{d^k}{dt^k} I(t) = (-1)^k (\varphi,U(t)i^k[\ad_A^{k-1}(H),A]U(t)^*\psi) = (-1)^k
(\varphi,U(t)i^k\ad_A^k(H)U(t)^*\psi).
\]
Taking $t=0$ implies (2). The computation here also establishes the formula connecting $\ad_A^j(H)$
and $H^{(j)}(0)$.

(2) $\Rightarrow$ (3): From the computation of $I$'s first derivative above, evaluated at $0$,
we observe that $[H,A]$ extends from the intersection domain to an $H$-bounded operator.
Hence this extension is unique, and indeed all the derivatives $H^{(j)}(0)$, $j=1,\dots, k$ 
are unique extensions by continuity. In particular $H^{(j)}(0)$ are symmetric operators on $\mathcal{D}(H)$ 
and, for $j=1,\dots,k$ and $\varphi, \psi\in\mathcal{D}(H)\cap\mathcal{D}(A)$,
\[
\frac{d^j}{dt^j} I(t) = (\varphi,U(t)i[A,H^{(j-1)}(0)]U(t)^*\psi) = (\varphi, U(t)H^{(j)}(0)U(t)^*\psi).
\]
That $\psi(t):=U(t)HU(t)^{*}\psi$ is itself continuous is a consequence of $U^{\circ}$ being a $C_0$-group on $\mathcal{D}(H)$. 
We assume inductively that $\psi(t)$
is $C^{k-1}([-1,1];\mathcal{H})$ and 
\[
\frac{d^{k-1}}{dt^{k-1}} \psi(t) = U(t) H^{(k-1)}(0)U(t)^*\psi.
\]
Assume now $\psi,\varphi\in\mathcal{D}(A)\cap \mathcal{D}(H)$ and compute 
\begin{align*}
& \frac1{t-s}\big((\varphi, \frac{d^{k-1}}{dt^{k-1}} \psi(t)) - (\varphi,\frac{d^{k-1}}{dt^{k-1}}\psi(s))\big)
-(\varphi,U(t)H^{(k)}(0)U(t)^*\psi) \\
& =\frac1{t-s}\int_s^t\big(\varphi,(U(r) H^{(k)}(0)U(r)^*-U(t)H^{(k)}(0)U(t)^*)\psi \big)dr.
\end{align*}
This identity now extends by continuity to $\varphi\in\mathcal{H}$ and $\psi\in\mathcal{D}(H)$.
We can furthermore estimate (for $s<t$)
\begin{align*}
& \big\| \frac1{t-s}\big(\frac{d^{k-1}}{dt^{k-1}} \psi(t)) - \frac{d^{k-1}}{dt^{k-1}}\psi(s)\big)
-U(t)H^{(k)}(0)U(t)^*\psi\big\|\\
& \leq 
\frac1{t-s} \int_s^t \big\|\big(U(r)H^{(k)}(0)U(r)^*-U(t)H^{(k)}(0)U(t)^*\big)\psi\big\| dr.
\end{align*}
That the right-hand side converges to zero when $s\to t$ (from the left) now follows from
the strong continuity of $U^{\circ}$ on $\mathcal{D}(H)$. A similar argument works for $s>t$.

(3) $\Rightarrow$ (1): Compute for $\varphi,\psi\in\mathcal{D}(H)\cap\mathcal{D}(A)$
\[
\frac{d^{j}}{d t^j} (\varphi,\psi(t))_{|t=0} = (\varphi,H^{(j)}(0)\psi). 
\]
Conversely one can compute the $j^{th}$ derivative in terms of iterated commutators,
and hence (1) follows. Note again, that the very first step in particular ensures that 
extensions are unique. 
\end{proof}

\subsection*{Examples}
\paragraph{\emph{1. N-body Schr{\"o}dinger operators.}}
Consider the operator
\begin{equation*}
	H=-\frac{1}{2}\Delta+\sum_{i<j}^{1,\dots,N}{V_{ij}(x_i-x_j)},
	\label{eq:NbodyH}
\end{equation*}with \emph{Coulomb pair potentials} $V_{ij}(x):=c_{ik}/(|x_i-x_j|)$, $c_{ik}\in\mathds{R}$, on $L^2(X)$, where
\begin{equation*}
X:=\left\{x=(x_1,\dots,x_N)\in\mathds{R}^{3N}|x_j\in\mathds{R}^3,\,\,1\leq j\leq N,\,\,\sum_{j=1}^{N}x_j=0\right\},
\end{equation*}\cite{HunzikerSigal2000b}. As a shorthand we write $x=(x_1,\dots,x_N)$. The unitary group of dilations, $U(\cdot)$ is defined by
\begin{equation*}
(U(t)\psi)(x):=e^{t\frac{3(N-1)}{2}}\psi\big(e^{t}x\big),
\end{equation*}and $U(t)=\exp(itA)$ for the generator of dilations $A$. From Proposition \ref{bdcommutatersandCk} infer for some $C>0$
\begin{equation*}
\|\ad_A^k(H)\|_{\mathfrak{B}(\mathcal{D}(p^2),\mathcal{H})}\leq C2^k.
\end{equation*}It is well known, that there is a Mourre estimate for a much more general class than the Coulomb N-body Hamiltonian, including the following example, \cite{HunzikerSigal2000b}. This enables Theorem \ref{analyticity}.\\Another example for $N$-body Schr\"{o}dinger operators to which Theorem \ref{analyticity} is applicable is defined with \emph{Yukawa pair potentials}. The pair potentials $V_{ik}$ are now given by
\begin{equation*}
V_{ij}(x):=\frac{c_{ik}e^{-\mu |x_i-x_j|}}{|x_i-x_j|},\,\,c_{ik}\in\mathds{R},\,\,\mu>0.
\end{equation*}Observe the estimate
\begin{equation*}
\left|\frac{d^k}{dt^k}\frac{e^{-t}}{r}e^{\mu r e^{t}} \right|_{\big|t=0}\leq k!a^{k},\,\,r:=|x_i-x_j|,
\end{equation*}for some $a>0$. The $r$-dependent functions on the right hand side of this inequality are infinitesimally $p^2$-bounded, which again shows the applicability of Theorem \ref{analyticity}. Hence non-threshold eigenvectors are analytic vectors with respect to $A$. This reproduces known results of \cite{BalslevCombes1971}.\\

\paragraph{\emph{2. The Spin-Boson Model.}}
The 'matter' Hamiltonian is defined as
\begin{equation*}
H_{\mathrm{at}}:=\epsilon\sigma_3,\,\,\epsilon>0,
\end{equation*}with the $2\times2$ Pauli-matrices $\sigma_1,\sigma_2,\sigma_3$. The corresponding Hilbert space is $\mathcal{H}_{\mathrm{at}}:=\mathds{C}^2$. We briefly list the definition of the quantised bosonic field, but for the details of \emph{second quantisation} we refer to \cite{DerezinskiGerard1999}. The Hilbert space of the bosonic field is the \emph{bosonic Fock space},
\begin{equation*}
\mathcal{F}_{+}:=\bigoplus_{n=0}^{\infty}\mathcal{S}_n\mathfrak{h}^{\otimes n},\,\,\mathfrak{h}:=L^2(\mathds{R}^3,d^3k),
\end{equation*}where $\mathcal{S}_n$ denotes the orthogonal projection onto the totally symmetric $n$-particle wave functions. We denote for $k\in\mathds{R}$ with $a(k)$ and $a^{\dagger}(k)$ the \emph{annihilation} and \emph{creation} operator, respectively. The energy of the free field, $H_{\mathrm{f}}$, is defined as
\begin{equation*}
H_{\mathrm{f}}=\int_{\mathds{R}^3} a^{\dagger}(k)\omega(k)a(k)d^3k,\,\,\omega(k):=\sqrt{k^2+m^2},\,\,m>0.
\end{equation*}The Hilbert space of the compound system is
\begin{equation*}
\mathcal{H}:=\mathcal{H}_{\mathrm{at}}\otimes \mathcal{F}.
\end{equation*}We define the coupling between atom and field by
\begin{equation*}
\Phi(v):=\frac{1}{\sqrt{2}}\int_{\mathds{R}^3}v(k)\{G\otimes a^{\dagger}(k)+G^{*}\otimes a(k)\}d^3k,
\end{equation*}with a complex $2\times2$ matrix $G$. The function $v$ is given by
\begin{equation*}
v(k):=\frac{e^{-\frac{k^2}{\Lambda^2}}}{\omega(k)^{\frac{1}{2}}},\,\,\forall k\in\mathds{R}^3.
\end{equation*}The constant $\Lambda>0$ plays the role of an ultraviolet cutoff. We define the Hamiltonian of the compound system, $H$, as
\begin{equation*}
H:=H_{\mathrm{at}}\otimes\mathbb{1}+\mathbb{1}\otimes H_{\mathrm{f}}+\Phi(v).
\end{equation*}Define,
\begin{equation*}
\alpha:=\frac{i}{2}\left(\nabla_{k}\cdot k+k\cdot\nabla_{k}\right).
\end{equation*}This operator is symmetric and densely defined on $L^2(\mathds{R}^{3})$ as it is the well known generator of the strongly continuous unitary group
\begin{equation*}
\left(u(t)\psi\right)(k):=e^{-\frac{3}{2}t}\psi\left(e^{-t} k\right).
\end{equation*}We denote the second-quantised operators of $\alpha$ and $u(t)$ by $A:=d\Gamma(\alpha)$ and $U(t):=\Gamma(u(t))$, respectively. $A$ is the generator of the strongly continuous unitary group $U(t)$. Observe that
\begin{equation*}
i^{\ell}\ad_A^{\ell}(H)=d\Gamma(i^{\ell}\ad_{\alpha}^{\ell}(\omega))+(-1)^{\ell+1}\Phi\left((i\alpha)^{\ell}v\right)
\end{equation*}and
\begin{equation}
\|\Phi\left((i\alpha)^{\ell}v\right)(H_{\mathrm{f}}+\mathbb{1})^{-\frac{1}{2}}\|\leq\|\omega^{-\frac{1}{2}}(i\alpha)^{\ell}v\|_{L^2}.
\label{eq:Phiestimate}
\end{equation}Since $(i\alpha)^{\ell}v=\frac{d^{\ell}}{dt^{\ell}}\left(e^{i\alpha t}v\right)\big|_{t=0}$, we have to estimate the multiple derivatives. Consider the map
\begin{equation*}
\overline{B\left(0,\frac{\pi}{4}\right)}\ni z\mapsto \left(k^2e^{-2z}+m^2\right)^{\frac{1}{2}}=\omega\left(e^{-z}k\right),\,\,k\in\mathds{R}^{3},
\end{equation*}where $\overline{B\left(0,\frac{\pi}{4}\right)}$ denotes the closed ball of radius $\pi/4$, centered at $0$. Observe, that
\begin{equation}
\frac{m}{\sqrt{2}}\leq|\omega\left(e^{-z}k\right)|\leq e^{\frac{\pi}{4}}\omega(k)
\label{eq:holomorphicomega}
\end{equation}where the lower bound implies that $z\mapsto \omega\left(e^{-z}k\right)^{-\frac{1}{2}}$ is holomorphic in $B\left(0,\frac{\pi}{4}\right)$, for all $k\in\mathds{R}^3$. The upper bound ensures that $\mathcal{D}(\mathbb{1}\otimes H_{\mathrm{f}})$ is b-stable with respect to $U(\cdot)$. Below, we will also show that $\ad_A(H)\in\mathfrak{B}(\mathcal{D}(H),\mathcal{H})$, which implies by Proposition \ref{bdcommutatersandCk} that $H\in C^1(A)$. Analogously we define the holomorphic map
\begin{equation*}
\overline{B\left(0,\frac{\pi}{4}\right)}\ni z\mapsto \frac{e^{-e^{-z}\frac{k^2}{\Lambda^2}}}{\omega(e^{-z}k)^{\frac{1}{2}}}=v\left(e^{-z}k\right),\,\,k\in\mathds{R}^{3}.
\end{equation*}We may compute by Cauchy's formula,
\begin{equation*}
\frac{d^{\ell}}{dz^{\ell}}\left(v(e^{-z}k)e^{-\frac{3}{2}z}\right)_{\big|_{z=0}}=\frac{\ell!\left(\frac{\pi}{4}\right)^{-\ell}}{2\pi}\intop_0^{2\pi}e^{-\frac{3}{2}\gamma(\varphi)}v\left(e^{-\gamma(\varphi)}k\right)e^{-i\ell\varphi}d\varphi,
\end{equation*}$\gamma(\varphi):=(\pi/4)e^{i\varphi}$, $\varphi\in[0,2\pi)$. Using the estimate
\begin{equation*}
\left|\frac{d^{\ell}}{dz^{\ell}}(v(e^{-z}k)e^{-\frac{3}{2}z})_{\big|_{z=0}}\right|\leq\left(\frac{m}{\sqrt{2}}\right)^{-\frac{1}{2}}e^{\frac{3\pi}{8}}e^{-e^{-\frac{\pi}{2}}\frac{k^2}{\Lambda^2}}\ell!\left(\frac{\pi}{4}\right)^{-\ell},\,\,\forall k\in\mathds{R}^{3},
\end{equation*}one finds together with (\ref{eq:Phiestimate})
\begin{equation*}
\|\Phi\left((i\alpha)^{\ell}v\right)(H_f+\mathbb{1})^{-\frac{1}{2}}\|\leq\ell!R^{-\ell},
\end{equation*}for some $R>0$. Analogously, we get from (\ref{eq:holomorphicomega})
\begin{equation*}
\left|\frac{d^{\ell}}{dz^{\ell}}(\omega(e^{-z}k))_{\big|_{z=0}}\right|\leq\ell!\left(\frac{\pi}{4}\right)^{-\ell}e^{\frac{\pi}{4}}\omega(k),
\end{equation*}so that
\begin{equation*}
\left\|d\Gamma(i^{\ell}\ad_{\alpha}^{\ell}(\omega))(H_{\mathrm{f}}+\mathbb{1})^{-1}\right\|\leq\|i^{\ell}\ad_{\alpha}^{\ell}(\omega)\omega^{-1}\|_{\infty}\leq\ell!c^{-\ell},
\end{equation*}for some $c>0$. From \cite{DerezinskiGerard1999} we may infer a Mourre estimate for our model. Derezi\'nski and G\'erard use a different generator of dilations, namely
\begin{equation*}
\alpha_{\omega}:=\frac{i}{2}\big((\nabla_k\omega)(k)\cdot \nabla_k+\nabla_k\cdot (\nabla_k\omega)(k)\big).
\end{equation*}It is also possible to prove a Mourre estimate using their techniques if $\omega(k)$ is radially increasing, $\omega(k)>0$, $\forall k\in\mathds{R}^{3}$ and $0$ is the only critical point of $\omega$. Thus, we conclude by Theorem \ref{analyticity} and Proposition \ref{bdcommutatersandCk} that any eigenstate pertaining to an embedded non-threshold eigenvalue is an analytic vector with respect to $A$.

\section{Preliminaries}
In what follows, we need some regularisation techniques from operator theory. It is convenient to perform calculations involving multiple commutators by using the so-called \emph{Helffer-Sj\"ostrand functional calculus}. Part and parcel of this calculus are certain extensions of a subclass of the smooth functions on $\mathds{R}$, the \emph{almost analytic extensions}. The following proposition allows us to define such extensions.
\begin{prop}
\label{almostana}
Consider a family of continuous functions $(f_n)_{n\in\mathds{N}}\subset C^{\infty}(\mathds{R})$, for which there is an $m\in\mathds{R}$, such that $\langle x\rangle^{k-m}f_n^{(k)}$ is uniformly bounded for all $n\geq0$. There exists a family of functions $(\tilde{f}_n)_{n\in\mathds{N}}$, such that
\begin{enumerate}
	\item $\supp(\tilde{f}_n)\subset\{z\in\mathds{C}|\Re z\in\supp(f_n)\,\text{   and   }\,|\Im z|\leq\langle\Re z\rangle\}$.
	\item $|\bar{\partial}\tilde{f}_n(z)|\leq C_N\langle z\rangle^{m-N-1}|\Im z|^N$ for all $N\geq 0$.
\end{enumerate}The constant $C_N$ does not depend on $n$.
\end{prop}For a proof of this statement see \cite{Moeller2000a}.
\begin{remark}
We will call these extensions for \emph{almost analytic extensions}, because $\bar{\partial}\tilde{f}_n$ vanishes approaching the real axis.
\end{remark}Let $\varepsilon>0$. For any self-adjoint operator $L$ and any $f\in C^{\infty}(\mathds{R})$ with
\begin{equation}
\sup_{t\in\mathds{R}}|f^{(k)}(t)\langle t\rangle^{k+\varepsilon}|
\end{equation}we may define a bounded operator $f(L)$, by
\begin{equation}
f(L):=\frac{1}{2\pi i}\int_{\mathds{C}}\bar{\partial}\tilde{f}(z)(z-L)^{-1}dz\wedge d\bar{z}.
\label{eq:HelfferSjoestrand}
\end{equation}The integral on the right hand side converges in operator norm. It is well known, that this definition coincides with the operator defined by functional calculus. Concerning the class $\mathscr{B}$ however, we cannot directly apply this definition. Inspired by a construction in \cite{MoellerSkibsted2004} we consider the following instead.

\begin{lemma}
Let $r\in\mathscr{B}$. There is an almost analytic extension of $t\mapsto r(t)/t=:\rho(t)$, which satisfies due to Proposition \ref{almostana} the bounds
\begin{equation}
|\bar{\partial}\tilde{\rho}(z)|\leq C_N\langle z\rangle^{-N-2}|\Im(z)|^N.
\label{eq:rhotilde1}
\end{equation}
\end{lemma}
\begin{proof}
Since $r$ is real analytic around $0$ we observe
\begin{equation*}
\sup_{|t|\leq 1}\big|\rho^{(k)}(t)\langle t\rangle^{k+1}\big|<\infty.
\end{equation*}On the other hand, the Leibniz rule yields $r^{(k)}(t)=\rho^{(k)}(t)t+k\rho^{(k-1)}(t)$ and thus by induction
\begin{equation*}
\sup_{|t|\geq1}\left|\rho^{(k)}(t)\langle t\rangle^{k+1}\right|<\infty.
\end{equation*}
\end{proof}For any $r\in\mathscr{B}$, set $r_n(t):=nr(t/n)$, $\rho(t):=r(t)/t$, $\forall t\in\mathds{R}$ and define $r_n(A)$ by functional calculus. If we require $\overline{\tilde{\rho}(z)}=\tilde{\rho}(\bar{z})$ the well known formula
\begin{equation}
r_n(t)=\frac{1}{2\pi i}\int_{\mathds{C}}\bar{\partial}\tilde{\rho}(z)\frac{t}{z-\frac{t}{n}}dz\wedge d\bar{z}
\label{eq:r(t)rep}
\end{equation}may be recovered. Observe, that
\begin{equation}
\frac{t}{z-\frac{t}{n}}=-n\left(1-\frac{z}{z-\frac{t}{n}}\right).
\label{eq:r(t)rep2}
\end{equation}The first term on the right hand side is constant and vanishes when computing commutators. Although we cannot use the formula (\ref{eq:r(t)rep}) directly as a representation of $r_n(A)$ on $\mathcal{H}$, it is possible to use it on the domain of $A$; a fact which is useful in the next lemma.
\begin{lemma}
Let $B\in C^{1}(A)$, where $B\in\mathfrak{B}(\mathcal{H})$. For any $r\in\mathscr{B}$ we have
\begin{equation}
[B,r_n(A)]=r_n'(A)\ad_A(B)+R(r_n,B),
\label{eq:uniformtaylor}
\end{equation}with
\begin{equation}
R(r_n,B):=\frac{1}{n2\pi i}\int_{\mathds{C}}\bar{\partial}\tilde{\rho}(z)zJ_n^2(z)[\ad_A(B),A]J_n(z)dz\wedge d\bar{z},
\label{eq:uniformtaylor2}
\end{equation}where $J_n(z):=n(nz-A)^{-1}$ and the integral being norm convergent. Moreover, there is a $c>0$
\begin{equation}
\slim_{n\to\infty}R(r_n)=0,\text{ $\mathrm{and}$  }\|R(r_n,B)\|\leq c\|\ad_A(B)\|.
\label{eq:uniformtaylors0}
\end{equation}If $B\in C^2(A)$, we have for any $n\in\mathds{N}$ and some $\alpha,\beta>0$
\begin{equation}
\|AR(r_n,B)\|\leq \alpha\|\ad_A^2(B)\|,\,\|R(r_n,B)\|\leq\frac{\beta}{n}\|\ad_A^2(B)\|.
\label{eq:uniformtayor3}
\end{equation}In addition,
\begin{equation}
\slim_{n\to\infty}AR(r_n,B)=0.
\label{eq:uniformtaylors}
\end{equation}
\label{uniformtaylor}
\end{lemma}
\begin{proof}
Let first $B\in C^1(A)$. If we consider $[r_n(A),B]$ as a form on $D(A)\times D(A)$, the commutator may be represented using (\ref{eq:r(t)rep}) with $t$ replaced by $A$, more precisely for all $\psi,\phi\in D(A)$
\begin{equation*}
(\phi,[B,r_n(A)]\psi)=\frac{1}{2\pi i}\int_{\mathds{C}}\bar{\partial}\tilde{\rho}(z)\big\{(A\phi,J_n(z)B\psi)-(\phi,BJ_n(z)A\psi)\big\}dz\wedge d\bar{z}.
\end{equation*}Observe, that the sum in the integrand is by definition
\begin{equation*}
(A\phi,J_n(z)B\psi)-(\phi,BJ_n(z)A\psi)=(\phi,[AJ_n(z),B]\psi).
\end{equation*}But since $B\in C^1(A)$, we obtain using (\ref{eq:r(t)rep2})
\begin{eqnarray*}
(\phi,[AJ_n(z),B]\psi)&=&(\phi,[nzJ_n(z)]\psi)\\
&=&((\phi,zJ_n(z)\ad_A(B)J_n(z)\psi)\\
&=&(\phi,zJ_n^2(z)\ad_A(B)\psi)+(\phi,zJ_n^2(z)[\ad_A(B),A]J_n(z)\psi).
\end{eqnarray*}There is an almost analytic extension $\tilde{\rho}(z)$ such that
\begin{equation}
|\bar{\partial}\tilde{\rho}(z)|\frac{|y|+|x|}{|y|^2}\leq C_N|y|^{N-2}\langle z\rangle^{-N-2},
\label{eq:uniformtaylorproof}
\end{equation}with $z=x+iy$, $x,y\in\mathds{R}$. Choose $N=2$ and observe that the integral
\begin{equation*}
\frac{1}{2\pi i}\int_{\mathds{C}}\bar{\partial}\tilde{\rho}(z)zJ_n^2(z)dz\wedge d\bar{z}
\end{equation*}converges in norm. Moreover,
\begin{equation*}
|\bar{\partial}\tilde{\rho}(z)|\frac{|z|}{|y|^3}(|y|+|x|)\leq C_3\langle z\rangle^{-3}.
\end{equation*}Thus from $r'(t)=\rho(t)+\rho'(t)t$ we may infer that this integral equals $r'_n(A)$. Estimate (\ref{eq:uniformtaylorproof}) shows that the integral (\ref{eq:uniformtaylor2}) converges in norm. Since
\begin{equation}
\slim_{n\to\infty}\frac{A}{n}J_n(z)=0,
\label{eq:uniformtaylorproof2}
\end{equation}the Theorem of Dominated Convergence implies (\ref{eq:uniformtaylors0}).\\Let now $B\in C^2(A)$. Choose in (\ref{eq:rhotilde1}) $N=3$, replace in (\ref{eq:uniformtaylor2}) $[\ad_A(B),A]$ with $\ad_A^2(B)$ and observe that the integrand of $AR(g_n,h)(B)$ is point-wise bounded by a constant times $\langle z\rangle^{-3}$. The term $R(g_n,h)(B)$ is point wise bounded by a constant times $\langle z\rangle^{-4}$. Both functions are in $L^1(\mathds{R}^2)$ and hence the bounds follow. Equation (\ref{eq:uniformtaylors}) is a consequence of (\ref{eq:uniformtaylor2}), (\ref{eq:uniformtaylorproof2}) and an application of the Theorem of Dominated Convergence.
\end{proof}

\begin{lemma}
\label{multcomm}
Let $r\in\mathscr{B}$ and $k\in\mathds{N}$. If $B\in C^{k}(A)$, then
\begin{equation*}
\slim_{n\to\infty}\ad_{r_n}^k(B)=\ad_{A}^k(B).
\end{equation*}
\end{lemma}
\begin{proof}
For $k=1$ the statement follows from Lemma \ref{uniformtaylor}. Let $k\in\mathds{N}$ and assume
\begin{equation*}
\slim_{n\to\infty}\ad_{r_n}^{k-1}(B)=\ad_{A}^{k-1}(B).
\end{equation*}The first term on the right-hand side of
\begin{equation*}
\ad_{r_n}(\ad_{r_n}^{k-1}(B))=\ad_{r_n}^{k-1}(\ad_{r_n}(B))=r_n'\ad_{r_n}^{k-1}(\ad_A(B))+\ad_{r_n}^{k-1}(R(r_n,B))
\end{equation*}converges strongly by the induction hypothesis and Lemma \ref{uniformtaylor} since $\ad_A(B)\in C^{k-1}(A)$. $R(r_n,\ad_{r_n}^{k-1}(B))$ is a sum of two integrals:
\begin{eqnarray*}
\ad_{r_n}^{k-1}(R(r_n,B)))&=&\frac{1}{2\pi i}\int_{\mathds{C}}\bar{\partial}\tilde{\rho}(z)z\frac{A}{n}J_n^2(z)\ad^{k-1}_{r_n}(\ad_{A}(B))J_n(z)dz\wedge d\bar{z}\\&&
-\frac{1}{2\pi i}\int_{\mathds{C}}\bar{\partial}\tilde{\rho}(z)zJ_n^2(z)\ad^{k-1}_{r_n}(\ad_{A}(B))\frac{A}{n}J_n(z)dz\wedge d\bar{z}.
\end{eqnarray*}Observe, that
\begin{equation*}
\slim_{n\to\infty}\frac{A}{n}J_n(z)=\slim_{n\to\infty}A(nz-A)^{-1}=0.
\end{equation*}By the uniform boundedness principle, the integrands are strongly convergent and converge to the product of the strong limits. Lemma \ref{uniformtaylor} and the Theorem of Dominated Convergence imply that we may exchange integration with the strong limit $n\to\infty$.
\end{proof}We use of the following expansion formula for commutators.
\begin{lemma}
\label{multcomm1}
Let $K,L\in\mathfrak{B}(\mathcal{H})$. Then, for any $k\in\mathds{N}$,
\begin{equation}
[K,L^k]=\sum_{j=1}^{k}\binom{k}{j}L^{k-j}\ad_L^j(K).
\label{eq:multcomm1}
\end{equation}
\end{lemma}
It is convenient to regularise the operator $A$ such that we may use the Helffer-Sj\"ostrand calculus and have sufficient flexibility in the proof. Let $g\in C^{\infty}_c(\mathds{R},\mathds{R})$ such that
\begin{equation}
\label{eq:regg1}
g(t)=t\,\forall t\in[-1,1],\,g(t)=2\,\forall t\geq3,\,g(t)=-2\,\forall t\leq-3,\,g'\geq0,
\end{equation}and that $tg'(t)/g(t)$ has a smooth square root; clearly $g\in\mathscr{B}$. We set $g_n(t):=ng(t/n)$ and define $g_n(A)$ by functional calculus. Observe, that
\begin{equation}
n\mapsto g_n^2(t)
\end{equation}is monotonously increasing for all $t\in\mathds{R}$. Set $\gamma(t):=g(t)/t$, for the function $g$ defined in (\ref{eq:regg1}). We may pick an almost analytic extension of $\gamma$, denoted by $\tilde{\gamma}$, such that $\tilde{\gamma}$ satisfies, up to a possibly different constant $C_N$, the same bounds as $\tilde{\rho}$ in (\ref{eq:rhotilde1}).

\section{Finite Regularity of Eigenstates}
\begin{proof}[Proof of Theorem \ref{finitereg}]
Using the convention $A^0=\mathbb{1}$, the statement is correct for $k=0$. Let now be $k\in\mathds{N}$ and assume $\psi\in\mathcal{D}(A^{k-1})$.
The starting point for the proof is
\begin{equation}
0=(\psi,i[h,g_n^kg_mg_n^k]\psi),
\label{eq:trivialvirial1}
\end{equation}which may be rewritten as
\begin{equation}
0=(\psi_n^{(k)},i\ad_{g_m}(h)\psi^{(k)}_n)+2\Re(\psi,g_mi[h,g_n^k]\psi_n^{(k)})+2\Re(\psi,[i[h,g_n^k],g_m]\psi^{(k)}_n),
\label{eq:trivialvirial2}
\end{equation}where we introduced the notation $\psi_n^{(k)}:=g_n^k\psi$. We abbreviate
\begin{equation}
I_0(n,m):=(\psi_n^{(k)},i\ad_{g_m}(h)\psi^{(k)}_n),
\label{eq:I_0def}
\end{equation}
\begin{equation}
I_1(n,m):=2\Re(\psi,g_mi[h,g_n^k]\psi_n^{(k)})
\label{eq:I1def}
\end{equation}and
\begin{equation}
I_2(n,m):=2\Re(\psi,[i[h,g_n^k],g_m]\psi^{(k)}_n)=2\Re(\psi,i[[h,g_m],g_n^k]\psi^{(k)}_n).
\label{eq:I2def}
\end{equation}We organise the proof in three steps. In the first step we extract from $I_1$ a term $I_0'$ which is of a similar type as $I_0$. Then, starting with (\ref{eq:trivialvirial2}) upper bounds to $I_0$, $I_0'$ are established. Finally, using Mourre's estimate we find lower bounds to $I_0$, $I_0'$, from which we conclude $\psi\in\mathcal{D}(A^k)$.
\begin{itemize}
\item[\underline{\textit{Step 1.}}]
\end{itemize}By an application of Lemma (\ref{multcomm1}) we rewrite $I_1(n,m)$ as
\begin{eqnarray}
\nonumber
I_1(n,m)&=&2\Re\left(i\sum_{j=2}^{k}\binom{k}{j}E_1(j,k,n,m)\right)+2k\Re(i(\psi_n^{(k-1)},g_mR(g_n,h)\psi_n^{(k)}))\\&&+2k\Re(i(\psi_n^{(k-1)},g_mg'_n\ad_{A}(h)\psi_n^{(k)})),\label{eq:I1}
\end{eqnarray}where $E_1(j,k,n,m):=(\psi_{n}^{(k-j)},g_m\ad_{g_n}^{j}(h)\psi^{(k)}_n)$ and $2k\Re(i(\psi_n^{(k-1)},g_mR(g_n,h)\psi_n^{(k)}))$ are present if $k\geq2$ only, in which case $\psi\in\mathcal{D}(A)$ by induction hypothesis. We discuss the term in the last line of (\ref{eq:I1}) first. One computes
\begin{eqnarray*}
2k\Re(i(\psi_n^{(k-1)},g_mg'_n\ad_{A}(h)\psi_n^{(k)}))&=&2k\Re(i(\psi_n^{(k)},\gamma_mp_n^2\ad_{A}(h)\psi_n^{(k)}))\\
&=&2k\Re(i(\psi_n^{(k)},\gamma_mp_n\ad_{A}(h)p_n\psi_n^{(k)}))\\&&+2k\Re(i(\psi_n^{(k)},\gamma_mp_n[p_n,\ad_{A}(h)]\psi_n^{(k)})),
\end{eqnarray*}with $\gamma_m$ being the operator $\gamma_m(A)$ and
\begin{equation*}
p(t):=\sqrt{\frac{tg'(t)}{g(t)}},\,p_n(t):=p(t/n).
\end{equation*}Hence, with
\begin{equation*}
E_1(j,k,n):=\lim_{m\to\infty}E_1(j,k,n,m)=(A\psi_{n}^{(k-j)},\ad_{g_n}^{j}(h)\psi^{(k)}_n),\,\,k\geq j\geq2,
\end{equation*}we obtain
\begin{eqnarray}
\nonumber
I_1(n)&:=&\lim_{m\to\infty}I_1(n,m)\\\nonumber&=&2\Re\left(i\sum_{j=2}^{k}\binom{k}{j}E_1(j,k,n)\right)+2k\Re(i(\psi_n^{(k-1)},AR(g_n,h)\psi_n^{(k)}))\\&&+2k\Re(i(\psi_n^{(k)},p_n[p_n,\ad_{A}(h)]\psi_n^{(k)}))+2k(\psi_n^{(k)},p_ni\ad_{A}(h)p_n\psi_n^{(k)}).
\label{eq:I1lim}
\end{eqnarray}Set
\begin{equation}
I_0'(n):=2k(\psi_n^{(k)},p_ni\ad_{A}(h)p_n\psi_n^{(k)}),\,\,I_1'(n):=I_1(n)-I_0'(n).
\label{eq:I_0'def}
\end{equation}
\begin{itemize}
\item[\underline{\textit{Step 2.}}]
\end{itemize}First note that by an application of Lemma \ref{multcomm}
\begin{eqnarray*}
I_2(n)&:=&\lim_{m\to\infty}I_2(n,m)=2\Re(\psi,i[\ad_A(h),g_n^k]\psi^{(k)}_n)\\
&=&2\Re\left(i\sum_{j=1}^{k}\binom{k}{j}E_2(j,k,n)\right),
\end{eqnarray*}with
\begin{equation*}
E_2(j,k,n):=(\psi_n^{(k-j)},\ad_{g_n}^j(\ad_{A}(h))\psi_n^{(k)}),\,\,k\geq j\geq1.
\end{equation*}Equation (\ref{eq:trivialvirial2}) may be rewritten as
\begin{equation}
I_0(n)+I_0'(n)=-I_1'(n)-I_2(n).
\label{eq:trivialvirial3}
\end{equation}In order to find an upper bound for the right hand side, we first estimate $E_1(j,k,n)$, $E_2(j,k,n)$ by
\begin{eqnarray*}
2|E_1(j,k,n)|&\leq&\epsilon_{jk}^{-1}\|\ad_{g_n}^{j}(h)g_n^{k-j}A\psi\|^2+\epsilon_{jk}\|\psi_n^{(k)}\|^2,\\
2|E_2(j,k,n)|&\leq&\mu_{jk}^{-1}\|\ad_{g_n}^{j}(\ad_A(h))\psi_n^{(k-j)}\|^2+\mu_{jk}\|\psi^{(k)}\|^2,
\end{eqnarray*}for all $\mu_{jk},\epsilon_{jk}>0$. The terms
\begin{equation*}
\|\ad_{g_n}^{j}(h)g_n^{k-j}A\psi\|,\,\,\|\ad_{g_n}^{j}(\ad_A(h))\psi_n^{(k-j)}\|
\end{equation*}are uniformly bounded in $n$ by Lemma \ref{multcomm}, $h\in C^{k+1}(A)$ and the induction hypothesis. For the remaining terms in (\ref{eq:I1lim}) we have
\begin{eqnarray*}
2k|(i(\psi_n^{(k-1)},AR(g_n,h)\psi_n^{(k)})|&\leq&k\left(\delta^{-1}\|R(g_n,h)A\psi^{(k-1)}\|^2+\delta\|\psi^{(k)}\|^2\right),\\
2k|(\psi_n^{(k)},p_n[p_n,\ad_{A}(h)]\psi_n^{(k)})|&\leq& k(\nu^{-1}\|[p_n,i\ad_A(h)]g_n\psi_n^{(k-1)}\|^2+\nu\|\psi_n^{(k)}\|^2).
\end{eqnarray*}$R(g_n,h)A$ is uniformly bounded in virtue of Lemma \ref{uniformtaylor}. The function $t\mapsto p(t)$ is by assumption smooth. Note that
\begin{eqnarray*}
[p_n,i\ad_A(h)]g_n=[p_n,i\ad_A(h)]A\gamma_n.
\end{eqnarray*}Further, since $p\in C^{\infty}_c(\mathds{R})$, an application of Proposition \ref{almostana} together with
\begin{equation*}
[p_n,\ad_A(h)]A=\frac{-1}{2\pi i}\int_{\mathds{C}}\bar{\partial}\tilde{p}(z)J_n(z)\ad_A^2(h)\frac{A}{n}J_n(z)dz\wedge d\bar{z}
\end{equation*}shows the uniform boundedness of $[p_n,\ad_A(h)]g_n$. For $1\leq j\leq k-1$ is $(\psi_n^{(j)})_{n\in\mathds{N}}$ convergent in norm to $A^j\psi$ and hence $(\|\psi_n^{j}\|)_{n\in\mathds{N}}$ is bounded. Choose now $\mu_{jk}:=\binom{k}{j}^{-1}k^{-1}C_0/12$, $\epsilon_{jk}:=\binom{k}{j}^{-1}(k-1)^{-1}C_0/12$, $\nu:=C_0/(12k)=:\delta$ and observe
\begin{equation}
I_0(n)+I_0'(n)-\frac{C_0}{3}\leq I_3(n),
\label{eq:upperbound}
\end{equation}where $(I_3(n))_{n\in\mathds{N}}$ is a bounded sequence.
\begin{itemize}
\item[\underline{\textit{Step 3.}}]
\end{itemize}Note, that we may assume $f_{\mathrm{loc}}(x)=\chi(h(x))$, $\forall x\in\mathds{R}$, for some compactly supported smooth function $\chi$ because $h$ is chosen to be invertible on the support of $f_{\mathrm{loc}}$. This implies $f_{\mathrm{loc}}(H)\in C^{k+1}(A)$, since $h\in C^{k+1}(A)$, see \cite[Prop.2.23]{GeorgescuGerardMoeller2004a}. Inserting the Mourre estimate from Condition \ref{condmourre1} yields
\begin{equation*}
(\psi_n^{(k)},i[h,A]\psi_n^{(k)})\geq C_0\|\psi_n^{(k)}\|^2-C_1\|f_{\mathrm{loc},\perp}\psi_n^{(k)}\|^2-(\psi_n^{(k)},K\psi_n^{(k)}).
\end{equation*}The second term is evaluated by
\begin{equation*}
f_{\mathrm{loc},\perp}g_n^k\psi=-\sum_{l=1}^{k}\binom{k}{l}(-1)^l\ad_{g_n}^l(f_{\mathrm{loc}})g_n^{k-l}\psi,
\end{equation*}where we used, that $\psi$ is an eigenstate and an adjoint version of (\ref{eq:multcomm1}). Thus, the contributions from this term are uniformly bounded in $n$ by Lemma \ref{multcomm} and the induction hypothesis. The spectral projection $\textbf{1}_{|A|\leq\Lambda}(A)$ defines a partition of unity, $\mathbb{1}=\textbf{1}_{|A|\leq\Lambda}(A)+\textbf{1}_{|A|>\Lambda}(A)$. Thus, we may write
\begin{equation*}
(\psi_n^{(k)},K\psi_n^{(k)})=(\psi_n^{(k)},\textbf{1}_{|A|\leq\Lambda}(A)K\psi_n^{(k)})+(\psi_n^{(k)},\textbf{1}_{|A|>\Lambda}(A)K\psi_n^{(k)}).
\end{equation*}Furthermore, we may estimate
\begin{equation*}
|(\psi_n^{(k)},\textbf{1}_{|A|\leq\Lambda}(A)K\psi_n^{(k)})|\leq\frac12\left(\frac{\|K\textbf{1}_{|A|\leq\Lambda}(A)\psi_n^{(k)}\|^2}{\nu}+\nu\|\psi_n^{(k)}\|^2\right)
\end{equation*}and
\begin{equation*}
|(\psi_n^{(k)},\textbf{1}_{|A|>\Lambda}(A)K\psi_n^{(k)})|\leq\frac12\left(\frac{\|\textbf{1}_{|A|>\Lambda}(A)K\|^2}{\delta}+\delta\right)\|\psi_n^{(k)}\|^2.
\end{equation*}Observe that since $K$ is compact and $\slim_{\Lambda\to\infty}\chi_{|A|>\Lambda}=0$ we have
\begin{equation*}
\forall\epsilon>0\,\exists\Lambda_{\epsilon}>0: \,\|\chi_{|A|>\Lambda_{\epsilon}}K\|<\epsilon,
\end{equation*}but this implies $\forall\Lambda\geq\Lambda_{\epsilon}$
\begin{equation*}
\|\textbf{1}_{|A|>\Lambda}(A)K\|=\|\textbf{1}_{|A|>\Lambda}(A)\textbf{1}_{|A|>\Lambda_{\epsilon}}(A)K\|\leq\epsilon.
\end{equation*}Thus, we may choose $\nu=C_0/9$, $\delta=C_0/9$ and pick then a $\Lambda>0$ big enough, such that
\begin{equation}
2\|\textbf{1}_{|A|>\Lambda}(A)K\|^2\leq C_0^2/(9)^2,
\label{eq:choiceofLambda}
\end{equation}i.e. $C_0-\nu-\delta-\epsilon=C_0/3$. Thus we arrive at
\begin{equation*}
I_0(n)+\frac{9\|K\textbf{1}_{|A|\leq\Lambda}(A)\psi_n^{(k)}\|^2}{2C_0}+C_1\left\|\sum_{l=1}^{k}\binom{k}{l}\ad_{g_n}^l(f_{\mathrm{loc}})g_n^{k-l}\psi\right\|^2\geq\frac{2C_0}{3}\|\psi_n^{k}\|^2.
\end{equation*}The left-hand side is bounded in $n$ by \textit{Step 2} and the induction hypothesis. Analogously, one finds for $I_0'(n)$
\begin{equation*}
I_0'(n)+b_n\geq\frac{C_0}{3}\|p_n\psi_n^{(k)}\|^2,
\end{equation*}for some $b_n\geq0$, $n\in\mathds{N}$ and $\sup_{n\in\mathds{N}}b_n<\infty$. Let
\begin{equation*}
I_4(n):=b_n+\frac{9\|K\textbf{1}_{|A|\leq\Lambda}(A)\psi_n^{(k)}\|^2}{2C_0}+C_1\left\|\sum_{l=1}^{k}\binom{k}{l}\ad_{g_n}^l(f_{\mathrm{loc}})g_n^{k-l}\psi\right\|^2.
\end{equation*}Finally, this gives with (\ref{eq:upperbound})
\begin{equation*}
\frac{C_0}{3}\left(\|p_n\psi_n^{(k)}\|^2+\|\psi_n^{(k)}\|^2\right)\leq I_3(n)+I_4(n),
\end{equation*}where the right-hand side is bounded in $n$. By definition of $g$ the result is now a consequence of the Theorem of Monotone Convergence applied to the left-hand side.
\end{proof}

\section{Eigenstates as analytic vectors}
\label{sectanalyticity}
To obtain explicit bounds, independent of the regularisations of $A$, we apply Lemma \ref{multcomm} and use (\ref{eq:trivialvirial3}) as a starting point.
\begin{prop}
Let $k\in\mathds{N}$, $h_{\lambda}(H)\in C^{k+1}(A)$ and Condition \ref{condmourre1} be satisfied. Then, for any eigenstate $\psi$ of $H$ with eigenvalue $\lambda\in\supp(f_{\mathrm{loc}})$ and $\Lambda\geq0$ being chosen as in (\ref{eq:choiceofLambda}) we have
\begin{eqnarray}
\nonumber
\|\psi^{(k)}\|^2&\leq&\frac{27\|K\textbf{1}_{|A|\leq\Lambda}(A)A^k\psi\|^2}{C_0^2}+\frac{6C_1}{C_0}\left\|\sum_{l=1}^{k}\binom{k}{l}\ad_{A}^l(f_{\mathrm{loc}})A^{k-l}\psi\right\|^2\\\nonumber&&
+\frac{96}{((1+2k)C_0)^2}\left(\|\ad_A^{k+1}(h)\psi\|^2+k^2\|\ad_A^2(h)A^{k-1}\psi\|^2\right)\\\nonumber&&
+\frac{12}{(1+2k)C_0}\sum_{j=2}^{k-1}\binom{k+1}{j+1}\left(|(A^{k+1-j}\psi,\ad_A^{j+1}(h)A^{k-1}\psi)|\right.\\&&+\left.|(A^{k-j}\psi,\ad_A^{j+2}(h)A^{k-1}\psi)|\right).
\label{eq:explicitbound1}
\end{eqnarray}
\end{prop}
\begin{remark}
The bounds derived in this proposition make the locally uniform boundedness of $A^k\psi$ in the sense of Condition 1.10 of \cite{FaupinMoellerSkibsted2010b} apparent.
\end{remark}
\begin{proof}
Note that $\psi\in \mathcal{D}(A^k)$ by Theorem \ref{finitereg}. We observe
\begin{equation*}
\lim_{n\to\infty}[p_n,\ad_A(h)]=\lim_{n\to\infty}\frac{-1}{n2\pi i}\int_{\mathds{C}}\bar{\partial}\tilde{p}(z)J_n(z)\ad_A^2(h)J_n(z)dz\wedge d\bar{z}=0,
\end{equation*}since $\bar{\partial}\tilde{p}$ has compact support and $h\in C^{k+1}(A)$. Further with $\psi^{(l)}:=A^l\psi$, for $0\leq l\leq k$,
\begin{eqnarray*}
\lim_{n\to\infty}E_1(j,k,n)&=&(\psi^{(k+1-j)},\ad_A^j(h)\psi^{(k)})=:E_1(j,k),\,\,k\geq j\geq2,\\
\lim_{n\to\infty}E_2(j,k,n)&=&(\psi^{(k-j)},\ad_A^{j+1}(h)\psi^{(k)})=:E_2(j,k),\,\,k\geq j\geq1.
\end{eqnarray*}Note that $E_1(j+1,k)=E_2(j,k)$ for $k-1\geq j\geq1$. Thus, equation (\ref{eq:trivialvirial3}) reads after taking the limit $n\to\infty$
\begin{eqnarray*}
(1+2k)(\psi^{(k)},i\ad_A(h)\psi^{(k)})=2\Re i\sum_{j=1}^{k-1}\binom{k+1}{j+1}E_2(j,k)+2\Re iE_2(k,k).
\end{eqnarray*}The term $E_2(k,k)$ is singular in the sense that one cannot commute one power of $A$ to the left-hand side and the estimate for $E_2(1,k)$ does not improve under such a manipulation. To estimate $E_2(1,k)$ we note
\begin{equation*}
-2\Re(\psi^{(k-1)},i\ad_A^2(h)\psi^{(k)})\leq\frac{1}{\epsilon}\|\ad_A^2(h)\psi^{(k-1)}\|^2+\epsilon\|\psi^{(k)}\|^2.
\end{equation*}We pick up a combinatorial factor $(k+1)k/2$ and thus choose
\begin{equation*}
\epsilon=\frac{(1+2k)C_0}{(k+1)k}2^{-3}.
\end{equation*}For $E_2(k,k)$, the combinatorial factor is $1$ and we estimate
\begin{equation*}
-2\Re(\psi,i\ad_A^{k+1}(h)\psi^{(k)})\leq\frac{1}{\mu}\|\ad_A^{k+1}(h)\psi\|^2+\mu\|\psi^{(k)}\|^2.
\end{equation*}Choose now
\begin{equation*}
\mu=(1+2k)C_02^{-4}.
\end{equation*}This gives with $(k+1)k/2\leq k^2$ the inequality
\begin{eqnarray*}
(\psi^{(k)},i\ad_A(h)\psi^{(k)})&-&C_02^{-3}\|\psi^{(k)}\|^2\leq\frac{2}{1+2k}\sum_{j=2}^{k-1}\binom{k+1}{j+1}|E_2(j,k)|\\&&+\frac{16}{(1+2k)^2C_0}\left(\|\ad_A^{k+1}(h)\psi\|^2+k^2\|\ad_A^2(h)\psi^{(k-1)}\|^2\right).
\end{eqnarray*}Note, that the upper bounds are modified as compared to the bounds in \textit{Step 2} of the proof of Theorem \ref{finitereg}. Namely we use for $2\leq j\leq k-1$,
\begin{equation*}
E_2(j,k)=(\psi^{(k+1-j)},\ad_A^{j+1}(h)\psi^{(k-1)})+(\psi^{(k-j)},\ad_A^{j+2}(h)\psi^{(k-1)}).
\end{equation*}Next, lower bounds are established using an analogous argument as in \textit{Step 3} of the proof of Theorem \ref{finitereg}. Observe that
\begin{eqnarray*}
(\psi^{(k)},&i\ad_A(h)&\psi^{(k)})+\frac{9\|K\textbf{1}_{|A|\leq\Lambda}(A)\psi^{(k)}\|^2}{2C_0}\\&&+C_1\left\|\sum_{l=1}^{k}\binom{k}{l}\ad_{A}^l(f_{\mathrm{loc}})\psi^{(k-l)}\right\|^2-C_02^{-3}\|\psi^{(k)}\|^2\geq\frac{C_0}{6}\|\psi^{(k)}\|^2.
\end{eqnarray*}Finally, we arrive at
\begin{eqnarray*}
\frac{C_0}{6}\|\psi^{(k)}\|^2&\leq&\frac{9\|K\textbf{1}_{|A|\leq\Lambda}(A)A^k\psi\|^2}{2C_0}+C_1\left\|\sum_{l=1}^{k}\binom{k}{l}\ad_{A}^l(f_{\mathrm{loc}})A^{k-l}\psi\right\|^2\\\nonumber&&
+\frac{16}{(1+2k)^2C_0}\left(\|\ad_A^{k+1}(h)\psi\|^2+k^2\|\ad_A^2(h)A^{k-1}\psi\|^2\right)\\\nonumber&&
+\frac{2}{1+2k}\sum_{j=2}^{k-1}\binom{k+1}{j+1}|E_2(j,k)|,
\end{eqnarray*}which implies (\ref{eq:explicitbound1}).
\end{proof}

\begin{lemma}
\label{Mortensformula}
Let $K,L\in\mathfrak{B}(\mathcal{H})$ and $J(z):=(z-K)^{-1}$ for $z\in\rho(K)$. Then,
\begin{equation}
\ad_L^k(J(z))=\sum_{a\in C(k)}\frac{k!}{a_1!\cdot\dots\cdot a_{n_a}!}J(z)\prod_{i=1}^{n_a}\ad_L^{a_i}(K)J(z),
\label{eq:Mortensformula}
\end{equation}where $C(k)$ denotes the set of all possible decompositions of $k=a_1+\dots+a_{n_a}$ in sums of natural numbers and further $a:=(a_1,\dots,a_{n_a})$.
\end{lemma}The formula may easily be observed to be correct. For a proof of similar statement see \cite{Rasmussen2010}.
\begin{proof}[Proof of Lemma \ref{lemmaCkproperty}]
We proof the statement by establishing the formula (\ref{eq:Mortensformula}) inductively for $K$ replaced by $H$ and $L$ replaced by $A$.
For $k=1$ we observe $\ad_A(J(z))=J(z)\ad_A(H)J(z)$, since $H\in C^1(A)$. Assume now for $k-1\in\mathds{N}$, $\rho(H)$,
\begin{equation}
\ad_A^{k-1}(J(z))=\sum_{a\in C(k-1)}\frac{(k-1)!}{a_1!\cdot\dots\cdot a_{n_a}!}J(z)\prod_{j=1}^{n_a}\ad_A^{a_j}(H)J(z).
\label{eq:Mortensformulaunbounded}
\end{equation}Observe, that $\ad_A^{a_j}(H)J(z)\in\mathfrak{B}(\mathcal{H})$, for all $1\leq j\leq n_a$. It is well known that the bounded elements in $C^1(A)$ form an algebra. This means that it suffices to check that each of the operators $\ad_A^{a_j}(H)J(z)$ is in $C^1(A)$. For $0\leq m\leq k-1$ we consider $[\ad_A^{m}(H)J(z),A]$. Let $\psi,\phi\in\mathcal{D}(A)\cap\mathcal{D}(H)$, then
\begin{eqnarray*}
(\psi,[\ad_A^{m}(H)J(z),A]\phi)&=&((-1)^{m}J(\bar{z})\ad_A^{m}(H)\psi,A\phi)\\&&+(A\psi,\ad_A^m(H)J(z)\phi)\\
&=&(\psi,[\ad_A^{m}(H),A]J(z)\phi)\\&&+((-1)^{m}\ad_A^{m}(A)\psi,J(z)\ad_A(H)J(z)\phi),
\end{eqnarray*}where in the last line we used
\begin{eqnarray*}
AJ(z)\psi=J(z)A\psi+J(z)\ad_A(H)J(z)\psi,\,\,\forall\psi\in\mathcal{D}(H).
\end{eqnarray*}By assumption, $[\ad_A^{m}(H),A]$ extends to a an element $\ad_A^{m+1}(H)\in\mathfrak{B}(\mathcal{D}(H),\mathcal{H})$, which implies that $[\ad_A^{m}(H)J(z),A]$ extends to a bounded operator for $0\leq m\leq k-1$, i.e. $\ad_A^m(H)J(z)\in C^1(A)$. Hence $H\in C^k(A)$.
\end{proof}

We devote the rest of this section to prove Theorem \ref{analyticity}.
\begin{proof}[Proof of Theorem \ref{analyticity}]
We organise the proof for analyticity in two steps and, for simplicity, we suppose the eigenvalue $\lambda$ with respect to $H,\psi$ is $0$. We consider $h(x):=x(1+\nu x^2)^{-1}$, for sufficiently small $\nu>0$, see Section \ref{mourre} and replace $f_{\mathrm{loc}}$ by $f_{\mathrm{ana}}$, defined in (\ref{eq:f}). By assumption and Section \ref{mourre}, this $h$ satisfies Condition \ref{condmourre1}. The first step consists of proving that $\psi$ is an analytic vector for $A$ under the condition
\begin{equation}
\|\ad_A^k(h)\|,\|\ad_A^k(f_{\mathrm{ana}})\|\leq k!w^{-k},\,\,\forall k\in\mathds{N},
\label{eq:condad2}
\end{equation}for some $w\in\mathds{R}_+$. In the second step we prove (\ref{eq:condad2}) using Condition \ref{condad}. Note, that it is sufficient to prove analyticity of the map $\theta\mapsto\exp(i\theta A)\psi=:\psi(\theta)$ in some ball around $0$. Namely, if $\psi(\cdot)$ is analytic in a ball then $\tilde{\psi}(t+\theta):=\exp(itA)\psi(\theta)$, $t\in\mathds{R}$ defines an analytic extension of this map to a strip. Alternatively, one observes the bounds in (\ref{eq:explicitbound1}) to be invariant under conjugation of $H$ with $\exp{itA}$, $t\in\mathds{R}$ and hence $\psi(\cdot)$ extends to an analytic function in a strip around the real axis.
\begin{itemize}
\item[\underline{\textit{Step 1.}}]
\end{itemize}Assume Condition (\ref{eq:condad2}) to be satisfied and abbreviate
\begin{eqnarray*}
\alpha(j,k)&:=&\frac{12}{(1+2k)C_0}\binom{k+1}{j+1}|(\psi^{(k+1-j)},\ad_A^{j+1}(h)\psi^{(k-1)})|,\\
\beta(j,k)&:=&\frac{12}{(1+2k)C_0}\binom{k+1}{j+1}|(\psi^{(k-j)},\ad_A^{j+2}(h)\psi^{(k-1)})|.
\end{eqnarray*}Motivated by Condition (\ref{eq:condad2}), we use the ansatz
\begin{equation*}
\|\psi^{(l)}\|\leq l!q^{-l},\,\,\mathrm{for}\,1\leq l\leq k-1,
\end{equation*}for some $q\in\mathds{R}_+$, $q<w$, independent of $l$. Employing the assumptions gives
\begin{equation*}
\alpha(j,k)\leq k!^2	q^{-2k}\frac{12}{C_0wk}(k+1-j)\left(\frac{q}{w}\right)^{j},
\end{equation*}thus
\begin{equation*}
(k!^2q^{-2k})^{-1}\sum_{j=2}^{k-1}\alpha(j,k)\leq\frac{12}{C_0w}\left(\frac{q}{w}\right)^{2}\sum_{j=0}^{k-3}\left(\frac{q}{w}\right)^{j}\leq\frac{12}{C_0w}\left(\frac{q}{w}\right)^{2}\frac{1}{1-\left(\frac{q}{w}\right)}.
\end{equation*}Analogously,
\begin{equation*}
\beta(j,k)\leq k!^2q^{-2k}\frac{12}{C_0wk}(j+2)\left(\frac{q}{w}\right)^{j+1}
\end{equation*}and consequently
\begin{equation*}
(k!^2q^{-2k})^{-1}\sum_{j=2}^{k-1}\beta(j,k)\leq\frac{24}{C_0w}\left(\frac{q}{w}\right)^{3}\sum_{j=0}^{k-3}\left(\frac{q}{w}\right)^{j}\leq\frac{24}{C_0w}\left(\frac{q}{w}\right)^{3}\frac{1}{1-\left(\frac{q}{w}\right)}.
\end{equation*}We continue by estimating, (\ref{eq:multcomm1}),
\begin{eqnarray*}
\left(\frac{6C_1}{C_0}\right)^{\frac12}\|f_{\mathrm{ana},\perp}\psi^{(k)}\|&\leq& \left(\frac{6C_1}{C_0}\right)^{\frac12}\sum_{j=1}^k\binom{k}{j}j!(k-j)!\left(\frac{q}{w}\right)^{j}q^{-k}\\&\leq& k!q^{-k}\left(\frac{6C_1}{C_0}\right)^{\frac12}\left(\frac{q}{w}\right)\frac{1}{1-\left(\frac{q}{w}\right)}.
\end{eqnarray*}Further,
\begin{equation*}
\frac{96k^2\|\ad_A^2(h)\psi^{(k-1)}\|^2}{C_0^2(1+2k)^2}\leq\frac{24}{C_0^2k^2w^2}\left(\frac{q}{w}\right)^2k!^2q^{-2k},
\end{equation*}
\begin{equation*}
\frac{96\|\ad_A^{k+1}(h)\psi\|^2}{C_0^2(1+2k)^2}\leq\frac{96}{C_0^2w^2}\left(\frac{q}{w}\right)^{2k}k!^2q^{-2k}
\end{equation*}and finally
\begin{equation*}
\frac{27}{C_0^2}\|K\textbf{1}_{|A|\leq\Lambda}(A)\psi^{(k)}\|^2\leq\frac{27\|K\|^2(\Lambda q)^{2k}}{C_0^2k!^2}k!^2q^{-2k}.
\end{equation*}Pick now $q$ sufficiently small, such that all pre-factors of $k!^2q^{-2k}$ are less than $1/6$ and observe that this can be done uniformly in $k$. Then, we obtain for our specified $q$
\begin{equation*}
\|\psi^{(k-1)}\|\leq (k-1)!q^{-(k-1)}\Longrightarrow\|\psi^{(k)}\|\leq k!q^{-k}.
\end{equation*}This proves that $\psi$ is an analytic vector for $A$, given Condition (\ref{eq:condad2}).
\begin{itemize}
\item[\underline{\textit{Step 2.}}]
\end{itemize}
We first compute the multiple commutators of $h$. For some $n_0\in\mathds{N}$, see Section \ref{mourre}, the function
\begin{equation*}
h(x)=-\frac{1}{2}((i-x/n_0)^{-1}+(-i-x/n_0)^{-1})
\end{equation*}and (\ref{eq:f}) satisfy Condition \ref{condmourre1}. It follows from Condition \ref{condad} and (\ref{eq:Mortensformulaunbounded}) in the proof of Lemma \ref{lemmaCkproperty} that the multiple commutators of $h$ may be expressed in terms of the multiple commutators of $J(z):=(z-H/n_0)^{-1}$,
\begin{equation}
\label{eq:multcommJ}
\ad_A^k(J(\pm i))=n_0^{-k}\sum_{a\in C(k)}\frac{k!}{a_1!\cdot\dots\cdot a_{n_a}!}J(\pm i)\prod_{i=1}^{n_a}\ad_A^{a_i}(H)J(\pm i),
\end{equation}for any $z$ in the resolvent set of $H$. The number of elements in $C(k)$ is given by $2^{k-1}-1$, which may be verified by induction. Thus, we may estimate (\ref{eq:multcommJ}) further in virtue of (\ref{eq:condad}).
\begin{eqnarray*}
\|\ad_A^k(J(\pm i))\|&\leq&k!v^{-k}(2^{k-1}-1)\leq k!w^{-k}\left(\frac{2w}{v}\right)^{k}.
\end{eqnarray*}Choose now $2w\leq v$ and conclude as in \emph{Step 1} by induction that for $h$, Condition \ref{condad} implies (\ref{eq:condad2}) and in particular, $h\in C^{\infty}(A)$. It is obvious that $f_{\mathrm{ana}}$ gives the same bounds, which completes the proof.
\end{proof}
\begin{remark}
\begin{enumerate}
\item If we had used $\arctan(x)$ instead of $h(x)=x(1+x^2)^{-1}$, we would have encountered the problem that the bounds (\ref{eq:condad2}) are easily obtained from (\ref{eq:condad}) in graph norm w.r.t. $H$, only. In contrast, the decay at infinity of our choice of $h$ allows naturally for bounds in operator norm.
\item Note, that the first step in the proof uses the relations (\ref{eq:condad2}) only and is, abstractly, independent of the stronger assumption (\ref{eq:condad}).
\end{enumerate}
\end{remark}

\section{The Mourre estimate in localised form}
\label{mourre}
The Mourre estimate is usually cast in a different form than it is used here. Let $H,A$ be self-adjoint operators, $H\in C^1(A)$. Let now $\tilde{C}_0>0$ and $\tilde{K}$ be a compact operator. We denote by $\textbf{1}_{I}(H)$ spectral projections of $H$ for an interval $I\subset\mathds{R}$. Suppose, that in the sense of quadratic forms on $\mathcal{H}\times\mathcal{H}$
\begin{equation}
\textbf{1}_{I}(H)i[H,A]\textbf{1}_{I}(H)\geq\tilde{C}_0\textbf{1}_{I}(H)-\tilde{K}.
\label{eq:mourre1}
\end{equation}This inequality is usually referred to as a \emph{Mourre estimate}. Choose $f_{\mathrm{loc}}\in C^{\infty}_c(\mathds{R})$ such that $\supp(f_{\mathrm{loc}}(H))\subseteq I$ and $f_{\mathrm{loc}}(\lambda)=1$. Set $f_{loc,\perp}:=1-f_{\mathrm{loc}}$. Then, multiplying (\ref{eq:mourre1}) from the left and the right with $f_{\mathrm{loc}}(H)$ yields
\begin{equation*}
f_{\mathrm{loc}}i[H,A]f_{\mathrm{loc}}\geq\tilde{C}_0+\tilde{C}_0f_{\mathrm{loc},\perp}^2-2\tilde{C}_0f_{loc,\perp}-K,
\end{equation*}where $K:=f_{\text{loc}}\tilde{K}f_{\text{loc}}$ is compact. As forms we observe $\forall\epsilon>0$
\begin{equation*}
2f_{\mathrm{loc},\perp}\leq\epsilon+\frac{1}{\epsilon}f_{loc,\perp}^2.
\end{equation*}Pick $\epsilon=1/4$. Therefore, we may rewrite (\ref{eq:mourre2}) as
\begin{equation}
\label{eq:peterpaul1}
f_{\mathrm{loc}}i[H,A]f_{\mathrm{loc}}\geq\tilde{C}_0\frac34-3\tilde{C}_0f_{loc,\perp}^2-K.
\end{equation}Let $h\in\mathscr{B}$. Set $h(t):=h(t-\lambda)$. By possibly shrinking the support of $f_{\mathrm{loc}}$ we may assume $\supp(f_{\mathrm{loc}})\subseteq\supp(h_{\lambda})$. To avoid obscuring the computations notationally, we refrain from writing $h_{\lambda}$ and use $h$ instead. Set $h_n(t):=nh(t/n)$, $\forall t\in\mathds{R}$  and abbreviate $K_n(z):=(z-H/n)^{-1}$. Then, by similar arguments as in Lemma \ref{uniformtaylor},
\begin{eqnarray}
f_{\mathrm{loc}}i\ad_A(h_n)f_{\mathrm{loc}}&=&f_{\mathrm{loc}}h'_ni\ad_A(H)f_{\mathrm{loc}}+R,
\label{eq:hntaylor}
\end{eqnarray}where
\begin{equation*}
R:=\frac{1}{2\pi n}\int_{\mathds{C}}\bar{\partial}\widetilde{\left(\frac{h}{t}\right)}(z)zK_n(z)^2f_{\mathrm{loc}}[\ad_A(H),H]f_{\mathrm{loc}}K_n(z)dz\wedge d\bar{z}.
\end{equation*}Note that
\begin{equation*}
f_{\mathrm{loc}}i\ad_A(H)f_{\mathrm{loc}}=f_{\mathrm{loc}}\mathbf{1}_I(H)i\ad_A(H)\mathbf{1}_I(H)f_{\mathrm{loc}}
\end{equation*}is a bounded operator on $\mathcal{H}$. Analogue estimates as in the proof of Lemma \ref{uniformtaylor} yield
\begin{equation*}
\|R\|\leq \frac{C}{n},
\end{equation*}for a $C\geq0$. This gives
\begin{eqnarray*}
\|f_{\mathrm{loc}}i\ad_A(H-h_n)f_{\mathrm{loc}}\|&\leq&\|(\mathbb{1}-h'_n)f_{\mathrm{loc}}i\ad_A(H)f_{\mathrm{loc}}\|+\frac{C}{n}\\
&\leq&C'\left(\|(\mathbb{1}-h'_n)\textbf{1}_{\supp(f_{\mathrm{loc}})}(H)\|+\frac{1}{n}\right),
\end{eqnarray*}for some $C'>0$. Taylor's theorem implies for positive $t\in\supp(f_{\mathrm{loc}})$
\begin{equation*}
|1-h_n'(t)|\leq\intop_0^{\frac{t}{n}}|h''(s)|ds\leq\frac{\sup_{t\in\supp(f_{\mathrm{loc}})}|t|}{n}\sup_{s\in\supp(f_{\mathrm{loc}})}|h''(s)|
\end{equation*}and analogously for negative $t\in\supp(f_{\mathrm{loc}})$. Thus, there is a $C''>0$ such that
\begin{equation*}
f_{\mathrm{loc}}i\ad_A(H-h_n)f_{\mathrm{loc}}\leq \frac{C''}{n}.
\end{equation*}
Choose $n_0\in\mathds{N}$ large enough such that
\begin{equation}
f_{\mathrm{loc}}i\ad_A(H-h_{n_0})f_{\mathrm{loc}}\leq \frac{\tilde{C}_0}{4}.
\label{eq:mourre2}
\end{equation}
Using $f_{loc,\perp}=1-f_{\mathrm{loc}}$ we obtain from (\ref{eq:mourre2}), (\ref{eq:peterpaul1})
\begin{equation}
i[h_{n_0},A]\geq\frac{\tilde{C}_0}{2}-3\tilde{C}_0f_{loc,\perp}^2-K-f_{loc,\perp}i[h_{n_0},A]f_{loc,\perp}-2\Re(f_{loc,\perp}i[h_{n_0},A]),
\label{eq:mourre3}
\end{equation}Note, that all operators appearing in (\ref{eq:mourre3}) are self-adjoint. With
\begin{eqnarray*}
f_{loc,\perp}i\ad_A(h_{n_0})f_{loc,\perp}&\leq&\|\ad_A(h_{n_0})\|f_{loc,\perp}^2,\\
\forall\delta>0:\hspace{1cm}\pm2\Re(f_{loc,\perp}i\ad_A(h_{n_0}))&\leq&\delta\|\ad_A(h_{n_0})\|^2+\frac{1}{\delta}f_{loc,\perp}^2,
\end{eqnarray*}and a choice of $\delta$ such that $\delta\|\ad_A(h_{n_0})\|^2\leq\tilde{C}_0/4$ we find
\begin{equation}
i[h_{n_0},A]\geq C_0-C_1f_{loc,\perp}^2-K,
\label{eq:mourre4}
\end{equation}where $0<C_0:=\tilde{C}_0/4$. The other constant is $C_1:=3\tilde{C}_0+\delta^{-1}+\|\ad_A(h_{n_0})\|$.

We may choose a $h$ which is real analytic and extends to an analytic function in a strip around the real axis. Thus it is possible to reformulate inequality (\ref{eq:mourre4}) using analytic functions only; a fact we rely on in the proof of our analyticity result.

Consider the real analytic function
\begin{equation}
f_{\mathrm{ana}}(x):=\frac{1}{1+(x-\lambda)^2}=\frac12\left(\frac{1}{1+i(x-\lambda)}+\frac{1}{1-i(x-\lambda)}\right),\forall x\in\mathds{R}.
\label{eq:f}
\end{equation}Replacing the constant $C_1$ with
\begin{equation*}
C_1\sup_{x\in\mathds{R}}\left(\frac{f_{\mathrm{loc},\perp,}(x)}{f_{\mathrm{ana},\perp}(x)}\right),
\end{equation*}where $f_{\mathrm{ana},\perp}:=1-f_{\mathrm{ana}}$, we may rewrite the Mourre estimate (\ref{eq:mourre4}) as
\begin{equation}
i[h,A]\geq C_0-C_1f_{\mathrm{ana},\perp}^2-K.
\label{eq:mourre5}
\end{equation}We denote the constant in front of $f_{\mathrm{ana},\perp}^2$ in a slight abuse of notation again with $C_1$.

\subsection*{Acknowledgements}
M. Westrich thanks Johannes-Gutenberg Universit{\"a}t Mainz, and V. Bach in particular, for support.  Moreover, both authors thank the Erwin Schr{\"o}dinger Institut (ESI) for hospitality, and in the case of M. Westrich for financial support in the form of a "Junior Research Fellowship".


\end{document}